\documentclass[twocolumn,10pt]{IEEEtran}
\usepackage{setspace}
\usepackage{color}
\usepackage{amsthm}
\usepackage{algorithm,algorithmic, amsbsy,amsmath,amssymb,epsfig,bbm,mathrsfs,fancyhdr,fancyvrb,subfigure,url}
\usepackage{epstopdf}
\usepackage{graphicx}
\usepackage{multirow} 

\newcommand{\beqn}{\begin{eqnarray}}
\newcommand{\eeqn}{\end{eqnarray}}

\newtheorem{theorem}{Theorem}
\newtheorem{lemma}{Lemma}
\newtheorem{proposition}{Proposition}

\newtheorem{claim}{Claim}
\newtheorem{remark}{Remark}

\title{Tier-Aware Resource Allocation in OFDMA Macrocell-Small Cell Networks}

\author{Amr Abdelnasser, Ekram Hossain, and Dong In Kim
\thanks{A. Abdelnasser (email: nasra@cc.umanitoba.ca) and E. Hossain (email: Ekram.Hossain@umanitoba.ca) are with the Department of Electrical and Computer Engineering at the University of Manitoba, Canada. D. I. Kim (email: dikim@skku.ac.kr) is with the  School of Information and Communication Engineering  at the Sungkyunkwan University (SKKU), Korea.}
\thanks{ A preliminary version of the paper~\cite{gc14-amr} has been submitted to IEEE Globecom'14.}
}

\begin{document}
\maketitle

\begin{abstract}

We present a joint sub-channel and power allocation framework for downlink transmission an orthogonal frequency-division multiple access (OFDMA)-based cellular network composed of a macrocell overlaid by small cells. In this framework, the resource allocation (RA) problems for both the macrocell and small cells are formulated as optimization problems. For the macrocell, we formulate an RA problem that is aware of the existence of the small cell tier. In this problem, the macrocell performs RA to satisfy the data rate requirements of macro user equipments (MUEs) while maximizing the tolerable interference from the small cell tier on its allocated sub-channels. Although the RA problem for the macrocell is shown to be a mixed integer nonlinear problem (MINLP), we prove that the macrocell can solve another alternate optimization problem that will yield the optimal solution with reduced complexity. For the small cells, following the same idea of tier-awareness, we formulate an optimization problem  that accounts for both RA and admission control (AC) and aims at maximizing the number of admitted users while simultaneously minimizing the consumed bandwidth. Similar to the macrocell optimization problem, the small cell problem is shown to be an MINLP. We obtain an upper bound on the optimal solution with reduced complexity through convex relaxation. In addition, we employ the dual decomposition technique to have a distributed solution for the small cell tier. Numerical results confirm the performance gains of our proposed RA formulation for the macrocell over the traditional resource allocation based on minimizing the transmission power. Besides, it is shown that the formulation based on convex relaxation yields a similar behavior to the MINLP formulation. Also, the distributed solution converges to the same solution obtained by solving the corresponding convex optimization problem in a centralized fashion.

{\em Keywords}:- Base station densification, small cells, OFDMA, downlink resource allocation, sub-channel and power allocation, admission control, convex optimization,  dual decomposition.

\end{abstract}

\section{Introduction}

As more and more customers subscribe to mobile broadband services, there is a tremendous growth in the demand for mobile broadband communications, together with the increased requirements for higher data rates, lower latencies and enhanced quality-of-service (QoS). Fueled by the popularity of smartphones and tablets with powerful multimedia capabilities, services and applications, it is anticipated that by 2020, the existing wireless systems will not be able to accommodate the expected $1000$-fold increase in total mobile broadband data \cite{Erxson}. Therefore, 5G cellular technologies are being sought. Of the several enabling technologies for 5G to handle the expected traffic demand, base station (BS) densification is considered as one of the most promising solutions \cite{Hossain2014}.

BS densification involves the deployment of a large number of low power BSs. This decreases the load per BS and leads to a better link between a user equipment (UE) and its serving BS owing to a smaller distance between them \cite{Bhushan2014}. This BS densification also creates a multi-tier network of nodes with different transmit powers, coverage areas and loads. Two issues arise in such a dense multi-tier network. The first issue is that the resource allocation (RA) in one tier cannot be done in isolation of the resource allocation in another tier. In other words, one tier should take into consideration the consequences of its RA decisions on the other tiers.  The second issue is that centralized RA solutions may not be feasible. Hence, there is a need for decentralized solutions for RA in different network tiers. 


In this paper, we formulate the RA problem for a two-tier orthogonal frequency division multiple access (OFDMA) wireless network composed of a macrocell overlaid by small cells. The objective of the macrocell is to allocate resources to its macro UEs (MUEs) to satisfy their data rate requirements. In addition, knowing about the existence of small cells,  the macrocell allocates the radio resources (i.e., sub-channel and power) to its MUEs in a way that can sustain the highest interference level from the small cells. For this reason, we formulate an optimization problem for the macrocell with an objective that is different from those in the traditional RA problems. Now, since small cells create dead zones around them in the downlink direction, the MUEs should be protected against transmission from the small cells \cite{Hossain2014}, \cite{Lopez2009}. Hence, knowing about the maximum allowable interference levels for MUEs, the small cells perform RA by solving an optimization problem whose objective function combines both the admission control (AC) and the consumed bandwidth (i.e., number of allocated sub-channels). The objective of the small cell tier is to admit as many small cell UEs (SUEs) as possible at their target data rates and consume the minimum amount of bandwidth. Again, this follows the same notion of tier-awareness by leaving as much bandwidth as possible for other network tiers (e.g., for device-to-device (D2D) communication). For this, an optimization problem is formulated for the small cell tier with the aforementioned objective, given the QoS requirements of SUEs and the interference constraints for the MUEs. Dual decomposition is used to have a decentralized RA and AC problem by decomposing the optimization problem into sub-problems for each small cell to solve. For this, only local channel gain information is used along with some coordination with the Home eNB Gateway (HeNB-GW) \cite{3gpp.36.300}. 

The key contributions of this paper can be summarized as follows:

\begin{itemize}

\item We develop a complete framework for tier-aware resource allocation in an OFDMA-based two-tier macrocell-small cell network with new objectives, which are different from the traditional sum-power or sum-rate objectives.

\item For the macrocell tier, we formulate a resource allocation problem that is aware of the existence of the small cell tier and show that it is a mixed integer nonlinear program (MINLP).

\item We prove that the macrocell can solve another alternate optimization problem that yields the optimal solution for the MINLP with polynomial time complexity.

\item We compare the proposed method for the macrocell RA problem to the traditional ``minimize the total sum-power" problem and show that the proposed method outperforms the traditional one in terms of the average number of admitted SUEs.

\item For the small cell tier, we formulate a joint resource allocation and admission control problem that aims at maximizing the number of admitted SUEs and minimizing their bandwidth consumption to accommodate additional tiers, and show that it is an MINLP.

\item We offer an upper bound solution to the MINLP through convex relaxation and propose a solution to the convex relaxation that can be implemented in a distributed fashion using dual decomposition.

\end{itemize}

The rest of this paper is organized as follows. Section \ref{sec:rel} reviews the related work.  Section \ref{sec:model} presents the system model and assumptions for this work. In Section \ref{sec:pbform}, the optimization problems are formulated for both the macrocell tier and the small cell tier, followed by the use of dual decomposition to have a decentralized operation. Numerical results are discussed in Section \ref{sec:numres} and finally Section \ref{sec:conc} concludes the work. A summary of the important symbols and notations used in the paper is given in Table \ref{table_sym}.

\begin{table*}[th]
\centering
\caption{Summary of the important symbols and notations}
\scriptsize
\begin{tabular}{ | l | l |}
\hline
Symbol & Description  \\
\hline 

$B$ & Index of the macrocell \\
$\boldsymbol{d}$ & Sub-gradient \\
$d^n$ & Element of the sub-gradient $\boldsymbol{d}$ \\
$\mathcal{F}$ & Set of all SUEs\\
$F$ & Number of SUEs \\
$f$ & Index of an SUE \\
$\mathcal{F}_s$ & Set of all SUEs served by small cell $s$\\
$g_{i,j}^{n}$ & Channel gain of the link between UE $j$ served by BS $i$ on sub-channel $n$ \\
$g$ & Lagrange dual function \\
$g_s$ & Lagrange dual function for small cell $s$\\
$I_{m}^{n}$ & Maximum tolerable interference level on sub-channel $n$ allocated to MUE $m$ \\
$I_{max}$ & Upper limit on the maximum tolerable interference level $I_m^n$\\
$I_{th}$ & Equal value for the maximum tolerable interference level $I_m^n~, \forall m \in \mathcal{M}, n \in \mathcal{N}$ \\
$I_{th, L}$ and $I_{th, H}$ & Lower and upper limits on $I_{th}$ in the bisection method \\
$I_{th, M}$ & Mean of $I_{th, L}$ and $I_{th, H}$ \\
$\mathcal{L}$ & Lagrangian function\\
$\mathcal{M}$ & Set of all MUEs \\
$M$ & Number of MUEs \\
$m$ & Index of an MUE \\
$\mathcal{N}$ & Set of all available sub-channels \\
$N$ & Number of sub-channels \\
$n$ & Index of a sub-channel \\
$N_{a,c}$ & Number of allocated sub-channels \\
$\mathcal{N}_m$ & Set of all sub-channels allocated to MUE $m$\\
$N_o$ & Noise power \\
$P_{i,j}^{n}$ & Power allocated to the link between UE $j$ served by BS $i$ on sub-channel $n$ \\
$\tilde{P}_{i,j}^{n}$ & Actual power allocated to the link between UE $j$ served by BS $i$ on sub-channel $n$ \\
$P_{B,max}$ & Total macrocell power \\
$P_{s,max}$ & Total small cell power \\
$R_B$ & Coverage radius of macrocell $B$\\
$R_m$ & Data rate requirement of MUE $m$ \\
$R_f$ & Data rate requirement of SUE $f$ \\
$\mathcal{S}$ & Set of all small cells \\
$S$ & Number of small cells \\
$s$ & Index of a small cell \\
$y_{s,f}$ & Admission control variable for SUE $f$ served by small cell $s$ \\
$\alpha^n$ & Scale factor for small cells transmission powers on sub-channel $n$ \\
$\gamma_{B,m}^{n}$ & Received SINR of an MUE $m$ served by macrocell B on sub-channel $n$ \\
$\gamma_{s,f}^{n}$ & Received SINR of an SUE $f$ served by small cell s on sub-channel $n$ \\
$\Gamma_{i,j}^{n}$ & Sub-channel allocation indicator for sub-channel $n$ allocated to UE $j$ served by BS $i$ \\
$\Delta f$ & Sub-channel bandwidth \\
$\delta$ & Termination tolerance in bisection method \\
$\boldsymbol{\eta}$ & Lagrange multiplier associated with the cross-tier interference \\
$\eta^{n}$ & Element of the Lagrange multiplier $\boldsymbol{\eta}$ \\
$\epsilon$ & Weighting factor \\

\hline
\end{tabular}
\label{table_sym}
\end{table*}

\section{Related Work}
\label{sec:rel}

The RA problem in OFDMA-based multi-tier cellular networks has been extensively studied in the literature. The authors in \cite{Kyuho2011} studied the RA problem in a multi-tier cellular network to maximize the sum-throughput subject to simple power budget and sub-channel allocation constraints. However, no QoS constraints were imposed. In \cite{Nguyen2014}, the RA problem in a femtocell network was modeled, with interference constraints for MUEs, in order to achieve fairness among femtocells. No QoS constraints, however, were imposed for femtocell users. In \cite{Guruacharya2013}, the RA problem in a two-tier macrocell-femtocell OFDMA network was modeled as a Stackelberg game, where the macrocell acts as the leader and the femtocells act as the followers. However, no  interference constraints for MUEs were considered. Also, no QoS constraints were imposed for femtocells. Reference \cite{Duy2014} studied the RA problem in a two-tier network composed of macrocells and femtocells which aimed at maximizing the sum-throughput of femtocells subject to total sum-rate constraint for the macrocell. Nevertheless, no QoS constraints were imposed for femtocells. The authors in \cite{Abdelnasser2014} studied the RA problem with QoS and interference constraints in a two-tier cellular network and used clustering as a technique to reduce the overall complexity. 

In the above works, either no QoS constraints were imposed or the RA problems with QoS constraints were assumed feasible. In other words, admission control \cite{Andersin1995}, which is a technique to deal with infeasibility when it is not possible to support all UEs with their target QoS requirements, was not studied. The authors in \cite{Lopez2014} proposed a distributed self-organizing RA scheme for a femtocell only network, with the aim of minimizing the total transmit power subject to QoS constraints. It was shown that minimizing the transmit power (which results in reduced interference) may improve throughput. 


Several works in the literature have considered the admission control problem. For cellular cognitive radio networks, \cite{Tadrous2011} studied the problem of admission and power control to admit the maximum number of secondary links and maximize their sum-throughput subject to QoS requirements and interference constraints for primary links. However, power control was done centrally. The authors in \cite{Mitliagkas2011} considered the problem of admission and power control, where the primary users are guaranteed a premium service rate and the secondary users are admitted (as many as possible) so long as the primary users are not affected. In \cite{Long2008}, the authors proposed a joint rate and power allocation scheme with explicit interference protection for primary users and QoS constraints for secondary users, where admission control was performed centrally. However, \cite{Tadrous2011}-\cite{Long2008} only considered single-channel systems. The authors in \cite{Shin2009} studied the problem of joint rate and power allocation with admission control in an OFDMA-based cognitive radio network subject to QoS requirements for secondary users and interference constraints for primary users. However, resource allocation and admission control were performed centrally. In addition, channels were randomly allocated to secondary users.

In relay networks, \cite{Phan2009} studied the problem of power allocation in amplify and forward wireless relay systems for different objectives, where admission control was employed as a first step preceding power control. However, only one channel was considered. In addition, power and admission control were done centrally. In \cite{Xiaowen2011}, a joint bandwidth and power allocation for wireless multi-user networks with admission control in relay networks was proposed for different system objectives. Unequal chunks of bandwidths were allocated. However, the resource allocation was performed centrally. 

For a two-tier small cell network, \cite{Siew2013} studied joint admission and power control. Small cells are admitted into the network so long as the QoS of macrocell users is not compromised. Admission and power control were performed in a distributed fashion. However, only a single channel system was considered. 
Reference \cite{Namal2010} proposed a distributed admission control mechanism for load balancing among sub-carriers with multiple QoS classes. In addition, small cells mitigate co-tier and cross-tier interferences using slot allocation of different traffic streams among different sub-carriers. However, no power allocation was performed. 

We notice that none of the quoted works considers the interaction between the different network tiers and the consequences of RA decisions of one tier on the other one. In addition, it is desirable to have an RA and AC scheme that is implementable in a distributed fashion in a dense multi-tier OFDMA network. Table \ref{table1} summarizes the  related work and their differences from the work presented in this paper. 


\begin{table*}[th]
\centering
\caption{Summary of Related Work}
\footnotesize
\begin{tabular}{ | c | c | c | c | c | c | c | c |}
\hline
Previous & Network type & Objective function & Multi-       & AC  & QoS               & Distributed & Impact of \\
works     &                        &                               &    channel &        & constraints    & solution &  one tier on \\
               &                       &                                &                 &         &                     &                & another \\
\hline \hline
\cite{Kyuho2011} & Two-tier macrocell/ & Maximize sum-rate & Yes & No & No & Yes & No \\
                              & small cell network   &for two tiers  &  &  &  &  & \\
\hline
\cite{Nguyen2014} &Two-tier macrocell/   &  Maximize sum-min &  Yes & No & For & Yes & No \\
                                &small cell network  & rate for small cells     &        &       & MUEs & & \\
\hline
\cite{Guruacharya2013} & Two-tier macrocell/ &Maximize sum-rate  & Yes & No & For  & Yes & No\\
                                        & small cell network   &for two tiers  &  &  &MUEs  & &  \\
\hline
\cite{Duy2014} &Two-tier macrocell/  & Maximize sum-rate & Yes & No & For  & Yes & No \\
                          & small cell network & for small cells       &  &  & MUEs & & \\
\hline
\cite{Abdelnasser2014} &Two-tier macrocell/  &  Maximize sum-rate & Yes & No & For &  Semi- & No\\
                                       & small cell network & for small cells &  &  & SUEs &  distributed &\\
\hline
\cite{Lopez2014} &Single-tier small  & Minimize sum-power &  Yes& No & Yes & Yes & No \\
                             &cell network  &  &  &  &  &  &\\
\hline
\cite{Tadrous2011} &Cognitive radio & Maximize number of links & No & Yes & Yes & No & No\\
                                & networks & with max sum-rate &  &  &  &  &\\
\hline
\cite{Mitliagkas2011} & Cognitive radio &Maximize number of users  & No & Yes & Yes & Yes & No\\
                                  & networks           &with min sum-power  &  &  &  &  &\\
\hline
\cite{Long2008} &  Cognitive radio &  Maximize min-rate and& No & Yes & Yes & No  & No\\
                           & networks           & maximize sum-log rate &  &  &  &  &\\
\hline
\cite{Shin2009} & Cognitive radio & Maximize sum-rate & Yes & Yes & Yes & No & No\\
                          & networks &  &  &  &  &  &\\
\hline
\cite{Phan2009} & Relay & Maximize min-SINR, min & No & Yes & Yes & No & No\\
                           & networks & max-power and  &  &  &  &  &\\
                           &                 & max sum-rate  &  &  &  &  &\\
\hline
\cite{Xiaowen2011} & Relay & Maximize sum-rate, max & Yes & Yes & Yes & No & No\\
                                & networks & min-rate and  &  &  &  &  &\\
                                &                 & minimize sum-power  &  &  &  &  &\\
\hline
\cite{Siew2013} & Two-tier macrocell/ & Minimize sum-power with& No & Yes & For MUEs & Yes & No\\
                           &small cell network  &  max number of SUEs &  &  & and SUEs &  &\\
\hline
\cite{Namal2010} & Two-tier macrocell/ & Maximize product of &Yes  & Yes & For SUEs & Yes & No\\
                             & small cell network & minimum of (2$\times$target rate &  &  &  &  &\\
		        &                             & - achieved rate) and achieved rate&  &  &  &  &\\
\hline
Our  & Two-tier macrocell/ & Maximize sum-tolerable & Yes &  Yes & For MUEs & Yes & Yes\\
            proposed    & small cell network  & interference for MUEs    &  &  &and SUEs  &  &\\
          scheme      &                              & and maximize admitted SUEs   &  &  &  &  &\\
                &                              &  with minimum bandwidth  &  &  &  &  &\\
\hline
\end{tabular}
\label{table1}
\end{table*}

\section{System Model, Assumptions, and Resource Allocation Framework}
\label{sec:model}

\subsection{System Model and Assumptions}


We consider the downlink of a two-tier cellular network, where a single macrocell, referred to by the index $B$ and with coverage radius $R_B$, is overlaid with $S$ small cells.
Denote by $\mathcal{S}$ the set of small cells, where $S=|\mathcal{S}|$. A closed-access scheme is assumed for all small cells, where access to a small cell is restricted only to the registered SUEs.  All small cells  are connected to the mobile core network. For example, femtocells can connect to the core network by using the DSL or CATV modems via an intermediate entity called the Femto Gateway (FGW) or HeNB-GW \cite{3gpp.36.300} which can take part in the resource allocation operation for femtocells. 

We denote by $\mathcal{M}$ the set of MUEs served by the macrocell $B$ with $M=|\mathcal{M}|$. Each MUE $m$ has a data rate requirement of $R_m$. In addition, denote by $\mathcal{F}$ the set of SUEs in the system with $F=|\mathcal{F}|$.  Each SUE $f$ has a data rate requirement of $R_f$. We refer to the set of SUEs served by small cell $s$ by $\mathcal{F}_s$. We assume that all UEs are already associated with their BSs and that this association remains fixed during the runtime of the resource allocation process. We have $\bigcup_{s=1}^{S}\mathcal{F}_s=\mathcal{F}$ and $\bigcap_{s=1}^{S}\mathcal{F}_s=\phi $. All MUEs exist outdoor and all SUEs exist indoor. We have an OFDMA system, where we denote by $\mathcal{N}$ the set of available sub-channels with $N=|\mathcal{N}|$ and $\Delta f$ is the bandwidth of a sub-channel $n$. Universal frequency reuse is assumed, where the macrocell and all the small cells have access to the set of sub-channels $\mathcal{N}$. $\Gamma _{i,j}^{n}$ is the sub-channel allocation indicator, i.e., $\Gamma _{i,j}^{n}=1$, if sub-channel $n$ is allocated to UE $j$ served by BS $i$ and takes the value of $0$ otherwise. 

The UEs are capable of using two modes of sub-channel allocation, namely, the exclusive mode and the time sharing mode.  For the exclusive mode, in a given transmission frame, sub-channel $n$ is used by one UE only. In the time sharing mode, a sub-channel $n$ is allocated to a certain UE a portion of the time. In this way, multiple UEs can time share a sub-channel $n$ in a given transmission frame \cite{Lataief99}. 

Denote by $P_{i,j}^{n}$ and $g _{i,j}^{n}$ the allocated power to and the channel gain of the link between BS $i$ and UE $j$ on sub-channel $n$. Channel gains are time varying and account for path-loss, log-normal shadowing, and fast fading. The channel gains are assumed to remain static during the resource allocation process. The received signal to interference plus noise ratio (SINR) $\gamma_{B,m}^{n}$ of an MUE $m$ served by macrocell $B$ on a sub-channel $n$ is defined as:

\begin{equation}
\label{sinr_macro}
\gamma_{B,m}^{n}=\frac{P_{B, m}^{n}g_{B, m}^{n}}{I_{m}^{n}+N_o}
\end{equation}
where $I_{m}^{n}$ is the maximum tolerable interference level at MUE $m$ on sub-channel $n$ and $N_o$ is the noise power. According to (\ref{sinr_macro}), the following constraint holds for small cell transmission powers on sub-channel $n$:

\begin{equation}
\Gamma _{B,m}^{n}\left(\sum_{s=1}^{S} \sum_{f \in \mathcal{F}_s}^{ } \Gamma _{s,f}^{n}P _{s,f}^{n}g_{s,m}^{n} \right ) \leq \Gamma _{B,m}^{n} I_{m}^{n}
\end{equation}
where the constraint is active only if sub-channel $n$ is allocated to MUE $m$, i.e., $\Gamma _{B,m}^{n}=1$. Similarly, we can define the received SINR $\gamma_{s,f}^{n}$ of an SUE $f$ served by small cell $s$ on a sub-channel $n$ as:

\begin{equation}
\label{sinr_femto}
\gamma_{s,f}^{n}=\frac{P_{s,f}^{n}g_{s, f}^{n}}{\sum_{m=1}^{M}\Gamma _{B,m}^{n}P_{B,m}^{n}g_{B, f}^{n}+N_o}.
\end{equation}

In (\ref{sinr_femto}), we consider cross-tier interference from macrocell $B$.\footnote{Note that it is straightforward to account for interference from other macrocells in (\ref{sinr_macro}) and (\ref{sinr_femto}). Nevertheless, since we are focusing on the interaction between RA decisions of the macrocell tier represented by macrocell $B$ and the overlaying small cells tier, interference from other macrocells will appear as a constant term in (\ref{sinr_macro}) and (\ref{sinr_femto}). Besides, $I_{m}^{n}$ will represent the maximum tolerable interference from the small cell tier.} On the other hand, co-tier interference from other small cells is assumed to be a part of the noise power $N_o$ due to the wall penetration loss and their relatively low transmission powers \cite{Zhang2012}.

\subsection{Tier-Aware Resource Allocation Framework}

Fig. \ref{RAflow} describes the RA framework proposed in this paper. Given the rate requirements for the MUEs, the macrocell starts by allocating resources to its MUEs and specifies the maximum tolerable interference levels on each allocated sub-channel. The macrocell then sends this RA information to the HeNB-GW which broadcasts it to the small cells. The small cells then perform RA and AC for its SUEs. For the resulting resource allocation for small cells, the MUEs perform interference measurements and report them to the macrocell BS. The macrocell BS then updates the HeNB-GW and the cycle repeats until the interference thresholds for all the MUEs are not violated and the RA and AC converge for all small cells. This cycle repeats due to the distributed nature of RA and AC in small cells. This repetition, however, does not take place if RA and AC in small cells are performed by a central controller. Note that the resource allocation in the macrocell from the first step remains fixed throughout the entire operation of the resource allocation process in the macrocell and small cells. The  awareness of the macrocell about the small cell tier is reflected in the way the radio resources are allocated in the macrocell. The macrocell allocates resources to its MUEs in a way that can tolerate the maximum interference possible from the samll cell tier. Note, however, that the minimum rate constraints of all MUEs must be satisfied in the sense that the rate requirement for none of the MUEs is compromised for admitting new SUEs. On the other hand, the awareness of the small cell tier about the existence of other tiers is reflected in the fact that the resource allocation in the small cell tier satisfies the rate requirements of the SUEs using the minimum amount of  bandwidth resources.

\begin{figure}[t]
\begin{center}
\includegraphics[width=3.5 in]{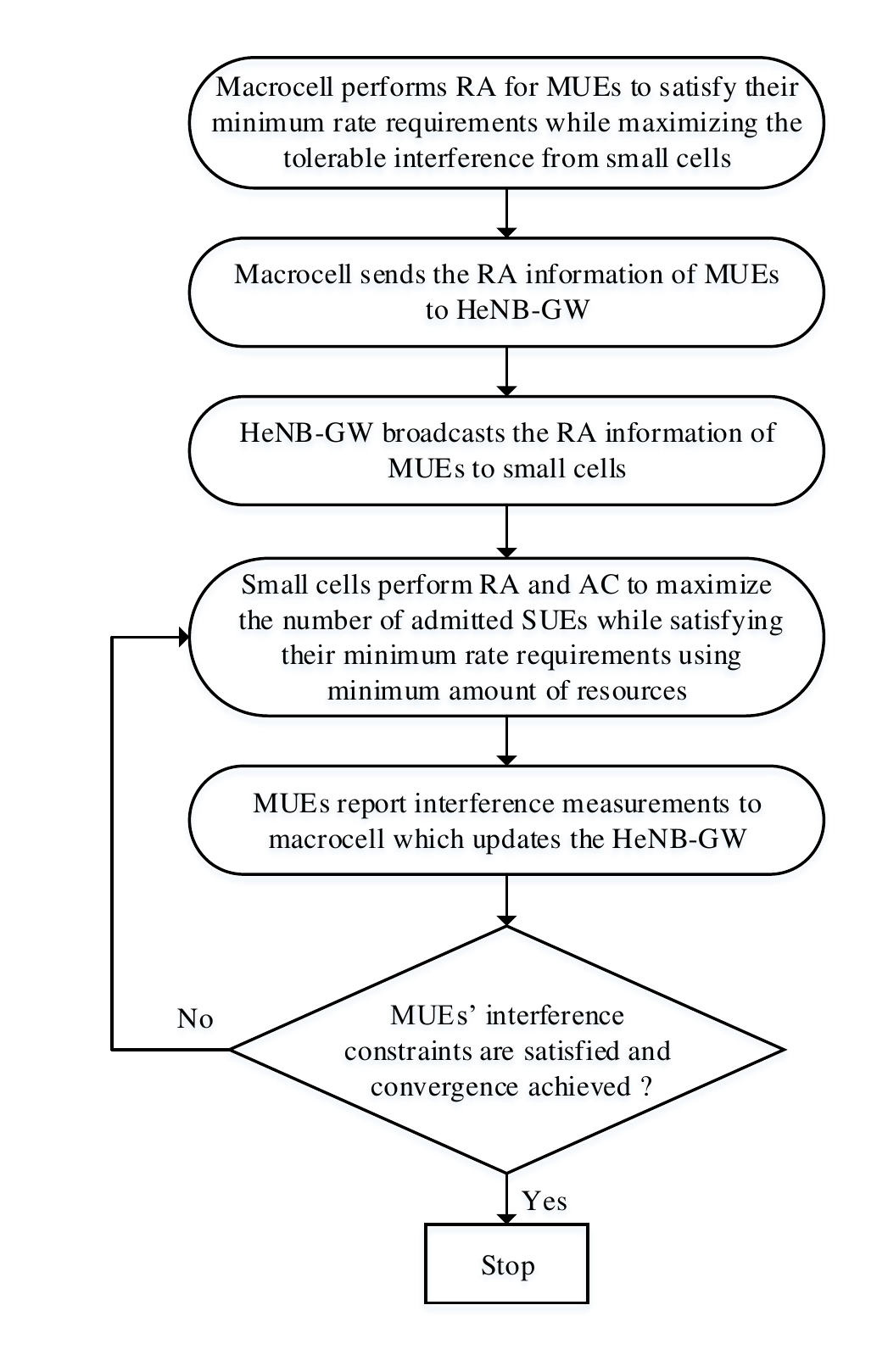}
\caption [c]{The RA framework for the macrocell and the small cells.}
\label{RAflow}
\end{center}
\end{figure}

\section{Problem Formulations for Resource Allocation}
\label{sec:pbform}

\subsection{Problem Formulation for Macrocell}

The macrocell is responsible for providing the basic coverage for the MUEs \cite{Astely2013}. Hence, the target of the macrocell is to allocate the resources to its MUEs to satisfy their data rate requirements and specify the maximum tolerable interference level by its MUEs on the allocated sub-channels. Different methods (i.e., corresponding to optimization problems for resource allocation with different objective functions), as will be shown later, can be followed to accomplish this task. It is of interest to study and understand the effects of different RA methods on other network tiers, which will be the small cell tier in our case.

\subsubsection{Maximize the sum of tolerable interference levels}

One way of performing RA in the macrocell is to allocate resources to the MUEs in a way that maximizes the sum of the maximum tolerable interference levels on the allocated sub-channels. The motivation behind this objective is to allow the maximum  possible freedom for the small cell tier in using the sub-channels. In this context, equal transmit power is assumed on the allocated sub-channels in the macrocell, i.e., $P_{B, m}^{n}=\frac{P_{B, max}}{N_{ac}}$ \cite{Hongjia2011}, where $P_{B, max}$ is the total macrocell power and $N_{ac} \leq N$ is the number of allocated sub-channels. Denote by $\mathcal{N}_m$ the set of sub-channels allocated to MUE $m$. Hence, we can define the optimization problem in (\ref{macro_pb3}), where the objective is to maximize the sum of the tolerable interference levels $I_{m}^{n}$ for all MUEs $m$ on all sub-channels $n$. C1 is the data rate constraint for each MUE $m$.  C2 and C3 indicate that the sets of sub-channels allocated to the MUEs are disjoint and constitute the entire set of sub-channels $\mathcal{N}$. C4 is a constraint added for numerical purposes, where $I_{max}$, is a very large number and $I_m^n=I_{max}$ means that sub-channel $n$ is not actually allocated to MUE $m$. Finally, C5 indicates that $I_{m}^{n}$ should be positive.

\begin{align}
\label{macro_pb3}
&\max_{\left \{ I_{m}^{n} \right \}}\sum_{m=1}^{M}\sum_{n\in \mathcal{N}_m}^{ }I_{m}^{n}  \nonumber \\
&{\text{subject}}~{\text{to}} \nonumber \\
& \mbox{C1}: \sum_{n \in \mathcal{N}_m}^{  }\Delta f \log_{2} \left(1+\frac{P_{B,m}^{n}g_{B,m}^{n}}{I_{m}^{n}+N_{o}} \right) \geq R_{m}, \forall m \in \mathcal{M}\nonumber \\
& \mbox{C2}: \mathcal{N}_i\bigcap_{ }^{ }\mathcal{N}_j=\varnothing, ~\forall i, j \in \mathcal{M} \nonumber \\
& \mbox{C3}: \bigcup_{m=1}^{M}\mathcal{N}_{m}\subseteq \left \{ 1,2,...,N \right \} \nonumber \\
& \mbox{C4}: I_{m}^{n}\leq I_{max}, \forall m \in \mathcal{M}, n \in \mathcal{N} \nonumber \\
& \mbox{C5}: I_{m}^{n} \geq 0, \forall m \in \mathcal{M}, n \in \mathcal{N}.
\end{align}

In general, (\ref{macro_pb3}) is an MINLP whose feasible set is non-convex due to C1 and the combinatorial nature of sub-channel allocation. Besides, $P_{B,m}^{n}$ is unknown as the number of allocated sub-channels $N_{ac}$ is not known yet. However, by carefully inspecting (\ref{macro_pb3}), some interesting features can be revealed which lead to the possibility of obtaining the optimal solution of (\ref{macro_pb3}) with polynomial time complexity. We shall assume first that (\ref{macro_pb3}) is always feasible and that in the extreme case, an MUE can have its rate requirement satisfied with one sub-channel only. The last assumption is possible thanks to the fact that the macrocell in our model has a control on the maximum interference level on the allocated sub-channel. The following Lemmas reveal some of the interesting features of (\ref{macro_pb3}).

\begin{lemma}
\label{lemma1}
At optimality, all data rate constraints C1 hold with equality.
\end{lemma}

\begin{proof}
Since the objective function in (\ref{macro_pb3}) is monotonically increasing in $I_m^n$ and C1 is monotonically decreasing in $I_m^n$, C1 must hold with equality at optimality for all MUEs. 
\end{proof}

\begin{lemma}
\label{lemma2}
At optimality, each MUE $m$ is assigned a single sub-channel $i$ with $I_m^i < I_{max}$.
\end{lemma}

\begin{proof}
To establish this result, we assume that for an MUE $m$ at optimality, $I_m^n < I_{max},~ \forall n \in \mathcal{N}_m$ with an objective function value  $\textup{Obj}_m=\sum_{n \in \mathcal{N}_m}^{ }I_m^n$ for MUE $m$. However, according to Lemma~\ref{lemma1}, the objective function is monotonically increasing in $I_m^n$, whereas the constraint C1 is monotonically decreasing in $I_m^n$. Therefore, we can decrease the value of $I_m^i$ on a certain sub-channel $i \in \mathcal{N}_m$ and increase the values of all other $I_m^j, j \in \mathcal{N}_m, j \neq i$. In this way, we end up with $I_m^j=I_{max}$, for $j \in \mathcal{N}_m, j \neq i$. Meanwhile, $I_m^i$ reaches a value $I_m^{i \ast}$ such that the rate constraint for MUE $m$ is met with equality resulting in a new objective function value ${\textup{Obj}}'_m=\left ( \left(|\mathcal{N}_m|-1\right)I_{max}+I_m^{i \ast} \right )$ which is clearly higher than $\textup{Obj}_m$. Hence, the initial assumption of optimality is contradicted.  
\end{proof}

Recall that, a value of $I_{max}$ for a certain $I_m^n$ means that sub-channel $n$ is not actually allocated to MUE $m$. This leads to the fact that a system with $N$ sub-channels and $M$ MUEs, where $M \leq N$, will end up with $M$ sub-channels only allocated to the $M$ MUEs which leads to a minimal use of the available system bandwidth. Hence, $P_{B, m}^{n}=\frac{P_{B, max}}{N_{ac}}=\frac{P_{B, max}}{M}$.

\begin{lemma}
\label{lemma3}
The allocated sub-channel $i$ for MUE $m$ is the one with the highest channel gain $g_{B,m}^{i}, i \in \mathcal{N}_m$.
\end{lemma}

\begin{proof}
According to Lemma \ref{lemma2}, at optimality, ${\textup{Obj}}'_m=\left ( \left(|\mathcal{N}_m|-1\right)I_{max}+I_m^{i \ast} \right )$ for MUE $m$, where $I_m^{i \ast}$ is selected such that the achieved data rate on sub-channel $i$ is equal to $R_m$. Hence, from the rate constraint formula, $I_m^{i \ast}=\left(\frac{P_{B,m}^{i}g_{B,m}^{i}}{2^{R_m/\Delta f}-1}-N_o \right)$. It is clear that $I_m^{i \ast}$ is directly proportional to $g_{B,m}^{i}$. Therefore, to maximize ${\textup{Obj}}'_m$, we need to maximize $I_m^{i \ast}$. Hence, MUE $m$ should be allocated sub-channel $i$ such that $i= \arg \underset{n \in \mathcal{N}_m}{\max }g_{B,m}^{n}$.  
\end{proof}

Based on the given Lemmas, we have the following Theorem.

\begin{theorem}
To maximize the sum of tolerable interference levels, the macrocell can solve the following alternate optimization problem:
\begin{align}
\label{macro_pb4}
&\max_{\left \{ \Gamma_{B,m}^{n} \right \}}\sum_{m=1}^{M}\sum_{n=1}^{N}\Gamma_{B,m}^{n}g_{B,m}^{n}  \nonumber \\
&{\textup{subject}}~{\textup{to}} \nonumber \\
& \mbox{C1}: \sum_{n=1}^{N}\Gamma_{B,m}^{n}=1,~\forall m \in \mathcal{M} \nonumber \\
& \mbox{C2}: \sum_{m=1}^{M}\Gamma_{B,m}^{n}\leq 1, ~\forall n \in \mathcal{N} \nonumber \\
& \mbox{C3}: \Gamma_{B,m}^{n}\in \left\{0,1 \right\},~\forall m \in \mathcal{M},n \in \mathcal{N} \nonumber \\
\end{align}
\end{theorem}
where the objective in (\ref{macro_pb4}) is to maximize the sum of the allocated sub-channel gains. C1 restricts the number of allocated sub-channels to any MUE $m$ to one sub-channel only, whereas C2 restricts sub-channel $n$ to be allocated to at most one MUE. Then for each MUE $m$ with allocated sub-channel $n$, the maximum tolerable interference level is given by: $I_m^{n}=\left(\frac{P_{B,m}^{n}g_{B,m}^{n}}{2^{R_m/\Delta f}-1}-N_o \right)$. Hence, (\ref{macro_pb3}) is solved optimally.

\begin{proof}
According to Lemmas \ref{lemma2} and \ref{lemma3}, at optimality, each MUE will have only one sub-channel which is the one with the highest gain.
Hence, we can define the optimization problem in (\ref{macro_pb4}) which is the well-known assignment problem that can be efficiently solved in polynomial time using the Hungarian method \cite{Wayne2003}. 
\end{proof}

In this way, the macrocell allocates sub-channels to its MUEs in a way that satisfies their data rate requirements and that can tolerate the maximum possible interference from the small cell tier.

\subsubsection{Minimize the total sum-power}

As a baseline, we consider another way of performing RA in the macrocell by minimizing the total sum-power of the macrocell given the data rate requirements of the MUEs. This problem has been studied extensively in the literature \cite{Kibeom2006}. However, the formulation developed in \cite{Kibeom2006} does not account for the maximum tolerable interference level $I_m^n$. Hence, we include it here with the required modification to determine the maximum tolerable interference level $I_m^n$. We have, thus, the following optimization problem:

\begin{align}
\label{macro_pb5}
&\min_{\left \{ P_{B,m}^{n} \right \}}\sum_{m=1}^{M}\sum_{n\in \mathcal{N}_m}^{ }P_{B,m}^{n}  \nonumber \\
&{\text{subject}}~{\text{to}} \nonumber \\
& \mbox{C1}: \sum_{n \in \mathcal{N}_m}^{  }\Delta f \log_{2} \left(1+\frac{P_{B,m}^{n}g_{B,m}^{n}}{I_{m}^{n}+N_{o}} \right) \geq R_{m}, \forall m \in \mathcal{M}\nonumber \\
& \mbox{C2}: \mathcal{N}_i\bigcap_{ }^{ }\mathcal{N}_j=\varnothing, ~\forall i, j \in \mathcal{M} \nonumber \\
& \mbox{C3}: \bigcup_{m=1}^{M}\mathcal{N}_{m}\subseteq \left \{ 1,2,...,N \right \} \nonumber \\
& \mbox{C4}: P_{B,m}^{n}\geq 0, \forall m \in \mathcal{M}, n \in \mathcal{N}. \nonumber\\
\end{align}

In (\ref{macro_pb5}), given the maximum tolerable interference level on each allocated sub-channel $I_{m}^{n}$, the macrocell seeks a power and sub-channel allocation solution that minimizes the sum-power. Although (\ref{macro_pb5}) is an MINLP whose feasible set is non-convex, it has been solved efficiently in \cite{Kibeom2006} in the dual domain using dual decomposition relying on the fact that the duality gap becomes virtually zero when the number of sub-channels in the system is sufficiently large. 

\begin{remark}
The reason for choosing the resource allocation solution based on the formulation in (\ref{macro_pb5}) as the  baseline is the following. With this solution, at optimality, all MUEs have their rate requirements satisfied with equality. This is also the case for the solution obtained from the formulation  in (\ref{macro_pb3}). In other words, with both the solutions, the MUEs achieve the same performance. The difference however lies in the way the resources are allocated, which subsequently  impacts  the performance of the small cell tier.

 
\end{remark}

In (\ref{macro_pb5}), the same value of $I_{m}^{n}$ is assumed $\forall m \in \mathcal{M}, n \in \mathcal{N}$, i.e., $I_{m}^{n}=I_{th}$. For a further fair comparison between (\ref{macro_pb3}) and (\ref{macro_pb5}), the macrocell adjusts the maximum tolerable interference level $I_{th}$ such that $\sum_{m=1}^{M}\sum_{n\in \mathcal{N}_m}^{ }P_{B,m}^{n}=P_{B, max}$. This can be accomplished by using the bisection method according to \textbf{Algorithm \ref{Alg_2}} as given below, where $I_{th,H} > I_{th,L}$.

\begin{algorithm}
\caption{Bisection method to find optimal $I_{th}$}
\begin{algorithmic}[1]
\STATE Macrocell initializes $I_{th,L}$,  $I_{th,H}$, and $\delta $
\WHILE{$|\sum_{m=1}^{M}\sum_{n\in \mathcal{N}_m}^{ }P_{B,m}^{n}-P_{B, max}|>\delta $}
\STATE $I_{th, M}=\left(I_{th,L}+I_{th,H}\right)/2$
\STATE Macrocell solves the optimization problem in (\ref{macro_pb5})
\IF{$\sum_{m=1}^{M}\sum_{n\in \mathcal{N}_m}^{ }P_{B,m}^{n}>P_{B, max}$}
\STATE $I_{th,H}=I_{th, M}$
\ELSIF{$\sum_{m=1}^{M}\sum_{n\in \mathcal{N}_m}^{ }P_{B,m}^{n}<P_{B, max}$}
\STATE $I_{th,L}=I_{th, M}$
\ENDIF
\ENDWHILE
\end{algorithmic}
 \label{Alg_2}
\end{algorithm}

After \textbf{Algorithm \ref{Alg_2}} terminates, $I_{th, M}$ gives the optimal value of  $I_{th}$. The optimization problem in (\ref{macro_pb5}) ends up with the power and sub-channel allocation to the MUEs with a uniform maximum tolerable interference level $I_{th}$ on all allocated sub-channels. In general, as will be shown in the numerical results, (\ref{macro_pb5}) leads to a higher number of allocated sub-channels to the MUEs than that (\ref{macro_pb3}) does. It is of interest to study the effect of the two different RA methods on the small cell tier.



\subsection{Problem Formulation for Small Cells}

Due to the small distance and the good channel conditions between small cells and SUEs, small cells are capable of serving registered SUEs with higher data rates than the macrocell. However, this should not be at the cost of QoS degradation at MUEs as they are served by the macrocell and provided with basic coverage at possibly lower rates \cite{Duy2014}. Hence, given the maximum tolerable interference levels on each allocated sub-channel for the MUEs, each small cell now tries to admit as many SUEs as possible at their target data rate by using the minimum possible bandwidth. Again, the idea here is to leave as much bandwidth as possible for the other network tiers (e.g., for device-to-device (D2D) communication).

\subsubsection{Centralized operation}

To accomplish the aforementioned requirements, we define the optimization problem in (\ref{femto_pb1}), where the objective function accounts for both admission control and sub-channel allocation. We have the admission control variable $y_{s,f}$ which takes the value of $1$ if SUE $f$ is admitted in small cell $s$ and $0$ otherwise. By controlling the weighting factor $\epsilon \in [0, 1]$, admission control can be given higher priority over the number of used sub-channels.

\begin{align}
\label{femto_pb1}
&\max_{\left \{ \Gamma _{s,f}^{n}, P_{s,f}^{n}, y_{s,f} \right \}}  \left(1-\epsilon \right )\sum_{s=1}^{S}\sum_{f \in \mathcal{F}_s}^{ }y_{s,f} - \epsilon \sum_{s=1}^{S}\sum_{f \in \mathcal{F}_s}^{ }\sum_{n=1}^{N}\Gamma _{s,f}^{n}  \nonumber \\
&{\text{subject}}~{\text{to}} \nonumber \\
& \mbox{C1}:\sum_{n=1}^{N}\Delta f \log_2\left(1+\frac{P_{s,f}^{n}g_{s,f}^{n}}{\sum_{m=1}^{M}\Gamma _{B,m}^{n}P_{B,m}^{n}g_{B,f}^{n}+N_o} \right )  \nonumber \\
& \hspace{120 pt} \geq y_{s,f} R_f,~\forall s \in \mathcal{S}, f \in \mathcal{F}_s \nonumber \\
& \mbox{C2}:\sum_{f \in \mathcal{F}_s}^{ }\sum_{n=1}^{N}P_{s,f}^{n} \leq P_{s, max},~\forall s \in \mathcal{S} \nonumber \\
& \mbox{C3}:\Gamma_{B,m}^{n}\left(\sum_{s=1}^{S}\sum_{f \in \mathcal{F}_s}^{ }P_{s,f}^{n}g_{s,m}^{n} \right )\leq \Gamma_{B,m}^{n}I_{m}^{n},~ \forall n \in \mathcal{N} \nonumber \\
& \mbox{C4}:P_{s,f}^{n} \leq \Gamma _{s,f}^{n}P_{s, max} ~\forall s \in \mathcal{S}, f \in \mathcal{F}_s, n \in \mathcal{N} \nonumber \\
& \mbox{C5}:\sum_{f \in \mathcal{F}_s}^{ }\Gamma _{s,f}^{n} \leq 1,~\forall s \in \mathcal{S}, n \in \mathcal{N} \nonumber \\
& \mbox{C6}:P_{s,f}^{n} \geq 0,~\forall s \in \mathcal{S}, f \in \mathcal{F}_s, n \in \mathcal{N} \nonumber \\
& \mbox{C7}:\Gamma _{s,f}^{n} \in \left\{0,1\right\},~\forall s \in \mathcal{S}, f \in \mathcal{F}_s, n \in \mathcal{N} \nonumber \\
& \mbox{C8}:y_{s,f} \in \left\{0,1\right\},~\forall s \in \mathcal{S}, f \in \mathcal{F}_s, n \in \mathcal{N}.
\end{align}

\begin{proposition}
\label{prop1}
By choosing $\epsilon < \frac{1}{1+SN}$, (\ref{femto_pb1}) admits the maximum number of SUEs while consuming the minimum number of sub-channels.
\end{proposition}

\begin{proof}
This Proposition can be proved in a way similar to that in \cite{Karipidis2008}. Let $\left(\Gamma _{s,f}^{n \ast}, P_{s,f}^{n \ast}, y_{s,f}^{\ast}\right)$, $\forall s \in \mathcal{S}, f \in \mathcal{F}_s, n \in \mathcal{N}$ denote an optimal solution of (\ref{femto_pb1}). Let $\left(\hat{\Gamma _{s,f}^{n }},\hat{ P_{s,f}^{n }}, \hat{y_{s,f}} \right)$, $\forall s \in \mathcal{S}, f \in \mathcal{F}_s, n \in \mathcal{N}$ be a feasible solution that admits one more SUE than the optimal solution, i.e., $\sum_{s=1}^{S}\sum_{f \in \mathcal{F}_s}^{ }\hat{y_{s,f}}=\sum_{s=1}^{S}\sum_{f \in \mathcal{F}_s}^{ }y_{s,f}^{\ast}+1$.
The objective of the feasible solution can be written as:
\begin{align}
& \left(1-\epsilon \right )\sum_{s=1}^{S}\sum_{f \in \mathcal{F}_s}^{ }\hat{y_{s,f}} - \epsilon \sum_{s=1}^{S}\sum_{f \in \mathcal{F}_s}^{ }\sum_{n=1}^{N}\hat{\Gamma _{s,f}^{n}} \overset{(1)}{\geq} \nonumber 
\end{align}
\begin{align}
& \left(1-\epsilon \right )\sum_{s=1}^{S}\sum_{f \in \mathcal{F}_s}^{ }y_{s,f}^{\ast}+\left(1-\epsilon \right )-\epsilon S N \overset{(2)}{\geq} \nonumber 
\end{align}
\begin{align}
& \left(1-\epsilon \right )\sum_{s=1}^{S}\sum_{f \in \mathcal{F}_s}^{ }y_{s,f}^{\ast} \overset{(3)}{\geq} \nonumber 
\end{align}
\begin{align}
& \left(1-\epsilon \right )\sum_{s=1}^{S}\sum_{f \in \mathcal{F}_s}^{ }y_{s,f}^{\ast} - \epsilon \sum_{s=1}^{S}\sum_{f \in \mathcal{F}_s}^{ }\sum_{n=1}^{N}\Gamma _{s,f}^{n \ast}. \nonumber 
\end{align}
The first inequality holds due to the fact that $\sum_{s=1}^{S}\sum_{f \in \mathcal{F}_s}^{ }\sum_{n=1}^{N}\hat{\Gamma _{s,f}^{n}}$ is upper bounded by $S N$ when all sub-channels in all small cells are allocated. The second inequality holds by setting $\left(1-\epsilon \right )-\epsilon S N > 0$. Hence, we have $\epsilon < \frac{1}{1+SN}$. The last inequality holds due to the non-negativity of $\sum_{s=1}^{S}\sum_{f \in \mathcal{F}_s}^{ }\sum_{n=1}^{N}\Gamma _{s,f}^{n \ast}$. In this way, the value of the objective function for the feasible solution is higher than the optimal one, which contradicts the optimality of $\left(\Gamma _{s,f}^{n \ast}, P_{s,f}^{n \ast}, y_{s,f}^{\ast}\right)$. Thus, there is no other solution that admits a higher number of SUEs under the constraint in (\ref{femto_pb1}). Given the optimum value for the admission control variable $y_{s,f}^{\ast}$, (\ref{femto_pb1}) reduces to a feasible sub-channel and power allocation problem with respect to the variables $\Gamma _{s,f}^{n}$ and  $P_{s,f}^{n}$ that aims at minimizing the number of used sub-channels subject to the given constraints.
\end{proof}

In (\ref{femto_pb1}), C1 is a data rate constraint for an SUE $f$ which is active only if SUE $f$ is admitted, i.e., $y_{s,f}=1$. C2 is the power budget constraint for each small cell $s$ restricting the total transmission power of small cell $s$ to be less than or equal to $P_{s, max}$. C3 is a constraint on the maximum cross-tier interference introduced to MUE $m$ using sub-channel $n$. C4 ensures that if sub-channel $n$ is not allocated to SUE $f$, its corresponding transmit power $P_{s,f}^{n}=0$. C5 constrains sub-channel $n$ to be allocated to at most one SUE $f$ in small cell $s$. C6 ensures that the power $P_{s,f}^{n}$ should be positive, and finally, C7 and C8 indicate that $\Gamma_{s,f}^{n}$ and $y_{s,f}$ are binary variables.

\begin{claim}
The optimization problem in (\ref{femto_pb1}) is always feasible.
\end{claim}

\begin{proof}
A trivial feasible solution of (\ref{femto_pb1}) is $\Gamma _{s,f}^{n}=0, P_{s,f}^{n}=0$ and $y_{s,f}=0,~\forall s \in \mathcal{S}, f \in \mathcal{F}_s, n \in \mathcal{N}$. 
\end{proof}

The problem in (\ref{femto_pb1}) is an MINLP whose feasible set is non-convex due to the combinatorial nature of sub-channel allocation and admission control. However, for small-sized problems, we use OPTI \cite{Jonathan2012}, which is a MATLAB toolbox to construct and solve linear, nonlinear, continuous and discrete optimization problems, to obtain the optimal solution. Obtaining the optimal solution, however, for larger problems is intractable. Another approach that can render the problem in (\ref{femto_pb1}) more tractable is to have a convex reformulation of (\ref{femto_pb1}) by relaxing the constraints C7 and C8 and allowing $\Gamma_{s,f}^{n}$ and $y_{s,f}$ to take any value in the range $[0, 1]$. Thus, $\Gamma_{s,f}^{n}$ is now a time sharing factor that indicates the portion of time sub-channel $n$ is allocated to SUE $f$ \cite{Zukang2005}, \cite{Tao2008}, whereas $y_{s,f}$ indicates the ratio of the achieved data rate for SUE $f$. Hence, we define the convex optimization problem in (\ref{femto_pbcnvx}), where $\tilde{P}_{s,f}^{n}$ can be related to $P_{s,f}^{n}$ in (\ref{femto_pb1}) as $\tilde{P}_{s,f}^{n}=\Gamma_{s,f}^{n} P_{s,f}^{n}$ to denote the actual transmit power \cite{Lataief99}. Now, (\ref{femto_pbcnvx}) is a convex optimization problem with a linear objective function and convex feasible set. It can be solved efficiently by the interior point method \cite{Boyd2004}.

Note that the system model assumed in (\ref{femto_pbcnvx}) differs from the original one in (\ref{femto_pb1}) as it allows time sharing of sub-channels and fractional satisfaction of the required data rates. Hence, the solution of (\ref{femto_pbcnvx}) gives an upper bound to the optimal solution of (\ref{femto_pb1}). However, it  helps by revealing some insights about the behavior of (\ref{femto_pb1}). Note that the solution of (\ref{femto_pbcnvx}) necessitates the existence of a central controller which can be, for example, the HeNB-GW. However, for dense small cell networks, having a decentralized solution with some coordination with a central entity will be a more viable option.

\begin{align}
\label{femto_pbcnvx}
&\max_{\left \{ \Gamma _{s,f}^{n}, \tilde{P}_{s,f}^{n}, y_{s,f} \right \}} \left(1-\epsilon \right )\sum_{s=1}^{S}\sum_{f \in \mathcal{F}_s}^{ }y_{s,f} - \epsilon \sum_{s=1}^{S}\sum_{f \in \mathcal{F}_s}^{ }\sum_{n=1}^{N}\Gamma _{s,f}^{n}  \nonumber \\
&{\text{subject}}~{\text{to}} \nonumber \\
& \mbox{C1}:\sum_{n=1}^{N}\Gamma _{s,f}^{n}\Delta f \log_2\left(1+\frac{\left(\tilde{P}_{s,f}^{n}g_{s,f}^{n}/\Gamma _{s,f}^{n} \right )}{\sum_{m=1}^{M}\Gamma _{B,m}^{n}P_{B,m}^{n}g_{B,f}^{n}+N_o} \right )   \nonumber \\
& \hspace{120 pt} \geq y_{s,f} R_f,~\forall s \in \mathcal{S}, f \in \mathcal{F}_s \nonumber \\
& \mbox{C2}:\sum_{f \in \mathcal{F}_s}^{ }\sum_{n=1}^{N}\tilde{P}_{s,f}^{n} \leq P_{s, max},~\forall s \in \mathcal{S} \nonumber \\
& \mbox{C3}:\Gamma_{B,m}^{n}\left(\sum_{s=1}^{S}\sum_{f \in \mathcal{F}_s}^{ }\tilde{P}_{s,f}^{n}g_{s,m}^{n} \right )\leq \Gamma_{B,m}^{n}I_{m}^{n},~ \forall n \in \mathcal{N} \nonumber \\
& \mbox{C4}:\sum_{f \in \mathcal{F}_s}^{ }\Gamma _{s,f}^{n} \leq 1,~\forall s \in \mathcal{S}, n \in \mathcal{N} \nonumber \\
& \mbox{C5}:\tilde{P}_{s,f}^{n} \geq 0,~\forall s \in \mathcal{S}, f \in \mathcal{F}_s, n \in \mathcal{N} \nonumber \\
& \mbox{C6}:\Gamma _{s,f}^{n} \in (0,1],~\forall s \in \mathcal{S}, f \in \mathcal{F}_s, n \in \mathcal{N} \nonumber \\
& \mbox{C7}:y_{s,f} \in \left[0,1\right],~\forall s \in \mathcal{S}, f \in \mathcal{F}_s, n \in \mathcal{N}. \nonumber \\
\end{align}

\subsubsection{Distributed operation}
\label{sec:DO}
To fulfill the requirement of having a decentralized solution for (\ref{femto_pbcnvx}), we use the dual decomposition method \cite{Boyd2007decomp}. For this purpose, we define the following partial Lagrangian function of the primal problem in (\ref{femto_pbcnvx}) formed by dualizing the constraint C3:

\begin{align}
\label{lag_femto_pb}
&\mathcal{L}\left(\Gamma _{s,f}^{n},\tilde{ P}_{s,f}^{n}, y_{s,f}, \boldsymbol{\eta}   \right ) \nonumber \\
&=\left(1-\epsilon \right )\sum_{s=1}^{S}\sum_{f \in \mathcal{F}_s}^{ }y_{s,f}-\epsilon \sum_{s=1}^{S}\sum_{f \in \mathcal{F}_s}^{ }\sum_{n=1}^{N}\Gamma _{s,f}^{n} \nonumber \\
&+ \sum_{n=1}^{N}\eta^{n}\left(\Gamma_{B,m}^{n}I_{m}^{n}-\Gamma_{B,m}^{n}\left(\sum_{s=1}^{S}\sum_{f \in \mathcal{F}_s}^{ }\tilde{P}_{s,f}^{n}g_{s,m}^{n} \right ) \right)
\end{align}
where $\boldsymbol{\eta}$ is the Lagrange multiplier vector (with elements $\eta^n$) associated with the cross-tier interference constraint C3. Then the Lagrange dual function is represented as

\begin{align}
\label{lagdulfn}
&g(\boldsymbol{\eta})=\max_{\left \{\Gamma _{s,f}^{n}, \tilde{P}_{s,f}^{n}, y_{s,f}  \right \}}\mathcal{L}\left(\Gamma _{s,f}^{n},\tilde{ P}_{s,f}^{n}, y_{s,f}, \boldsymbol{\eta}   \right ) \nonumber \\
&{\text{subject}}~{\text{to}} \nonumber \\
&\mbox{C1},~ \mbox{C2},~ \mbox{C4}-\mbox{C7}.
\end{align}

From (\ref{lag_femto_pb}), the maximization of $\mathcal{L}$ can be decomposed into $S$ independent optimization problems for each small cell $s$ as follows:
\begin{align}
\label{lagdulfns}
&g_s(\boldsymbol{\eta})=\max_{\left \{ \Gamma _{s,f}^{n}, \tilde{P}_{s,f}^{n}, y_{s,f} \right \}}\left(1-\epsilon \right )\sum_{f \in \mathcal{F}_s}^{ }y_{s,f}-\epsilon \sum_{f \in \mathcal{F}_s}^{ }\sum_{n=1}^{N}\Gamma _{s,f}^{n}\nonumber \\
&\hspace{85 pt}-\sum_{f \in \mathcal{F}_s}^{ }\sum_{n=1}^{N}\eta^{n}\Gamma_{B,m}^{n}\tilde{P}_{s,f}^{n}g_{s,m}^{n} \nonumber \\
&{\text{subject}}~{\text{to}} \nonumber \\
&\mbox{C1},~ \mbox{C2},~ \mbox{C4}-\mbox{C7}, \quad \forall s \in S.
\end{align}
Thus, the Lagrange dual function is
\begin{align}
&g(\boldsymbol{\eta})=\sum_{s=1}^{S}g_s(\boldsymbol{\eta})+\sum_{n=1}^{N}\eta^{n}\Gamma_{B,m}^{n}I_{m}^{n}.
\end{align}

Then, the dual problem is given by:
\begin{align}
\label{lagdualpb}
&\min_{\boldsymbol{\eta} \geq 0}g(\boldsymbol{\eta}). 
\end{align}

In order to solve the dual problem,  $\boldsymbol{\eta}$ can be updated efficiently  using the ellipsoid method \cite{Boyd2008ellipse}. A sub-gradient $\boldsymbol{d}$ of this problem required for the ellipsoid method is derived in the following proposition.

\begin{proposition}
\label{prop2}
For the optimization problem in (\ref{femto_pbcnvx}) with a dual objective $g(\boldsymbol{\eta})$ defined in (\ref{lagdulfn}), the following choice of $d^{n}$ is a sub-gradient for $g(\boldsymbol{\eta})$:
\begin{align}
\label{subgrad}
&d^{n}=\Gamma_{B,m}^{n}I_{m}^{n}-\Gamma_{B,m}^{n}\left(\sum_{s=1}^{S}\sum_{f \in \mathcal{F}_s}^{ }\tilde{P}_{s,f}^{n\ast }g_{s,m}^{n} \right )
\end{align}
\end{proposition}
where $d^{n}$ is an element of $\boldsymbol{d}$ and $\Gamma _{s,f}^{n\ast }, \tilde{P}_{s,f}^{n\ast }$, and $y_{s,f}^{\ast }$ optimize the maximization problem in the definition of $g(\boldsymbol{\eta})$.

\begin{proof}
For any $\boldsymbol{\xi}\geq 0$,
\begin{align}
&g(\boldsymbol{\xi} )\geq \mathcal{L}(\Gamma _{s,f}^{n\ast }, \tilde{P}_{s,f}^{n\ast }, y_{s,f}^{\ast }, \boldsymbol{\xi})\nonumber \\
&=g(\boldsymbol{\eta})+\sum_{n=1}^{N}\left(\xi^{n}-\eta^{n} \right )\left [\vphantom{\left(\sum_{s=1}^{S}\sum_{f \in \mathcal{F}_s}^{ }\tilde{P}_{s,f}^{n\ast }g_{s,m}^{n} \right )}\Gamma_{B,m}^{n}I_{m}^{n} \right. \nonumber \\
&\left.-\Gamma_{B,m}^{n}\left(\sum_{s=1}^{S}\sum_{f \in \mathcal{F}_s}^{ }\tilde{P}_{s,f}^{n\ast }g_{s,m}^{n} \right )  \right ]. \nonumber 
\end{align} 
\end{proof}

\textbf{Algorithm \ref{Alg_1}} gives a practical implementation of the distributed joint RA and AC operation for the small cells. After the macrocell has performed RA for its MUEs, it sends the sub-channel allocation information for its MUEs and the initialized multiplier $\boldsymbol{\eta}$ to the HeNB-GW. For a given $\boldsymbol{\eta}$, all small cells solve their optimization problem in (\ref{lagdulfns}) simultaneously. For the given resource allocation in the small cells, the MUEs estimate the resulting interference levels and send them to the macrocell which updates the multiplier values  using the ellipsoid method. The macrocell then informs  the updated multiplier values to the HeNB-GW, which broadcasts them to the small cells, and the entire operation repeats. Note that the small cells can obtain the channel gains $g_{s,m}^{n}$ relying on the techniques proposed in \cite{Kyuho2011}.

Finally, the remaining issue is to obtain a feasible primal solution to (\ref{femto_pbcnvx}) based on the resulting solution from the Lagrangian dual in (\ref{lagdualpb}). It has been reported in \cite{Boyd2007decomp} and \cite{Komodakis2011} that the iterations of the dual decomposition method are, in general, infeasible with respect to (\ref{femto_pbcnvx}). This infeasibility, however, is not severe as large constraint violations usually get penalized. Hence, using a simple procedure, one can recover a primal feasible solution that serves as a lower bound for the optimal solution of (\ref{femto_pbcnvx}). Suppose that the reported interference level by an MUE $m$ allocated a sub-channel $n$ was found to be:


\begin{align}
&\sum_{s=1}^{S} \sum_{f \in \mathcal{F}_s}^{ } \tilde{P}_{s,f}^{n}g_{s,m}^{n}  = \alpha^{n}  I_{m}^{n},~ \alpha^{n} > 1.
\end{align}

A straightforward way to recover feasibility is for the HeNB-GW to instruct all small cells transmitting on sub-channel $n$ to scale down their transmission powers by the factor $\alpha^{n}$. For the updated power values, the entire problem is solved to obtain the updated values of sub-channel allocation and admission control variables. The gap between the lower bound offered by this procedure and the upper bound offered by (\ref{lagdualpb}), referred to as the duality gap,  diminishes with iterations. Convergence to the optimal solution is guaranteed since the primal optimization problem in (\ref{femto_pbcnvx}) is convex.

\begin{algorithm}
\caption{Distributed joint RA and AC algorithm}
\begin{algorithmic}[1]
\STATE Macrocell initializes $\boldsymbol{\eta}$, $L_{max}$, sends sub-channel allocation information $\Gamma_{B,m}^{n}$ and $\boldsymbol{\eta}$ to HeNB-GW and sets iteration counter $l=1$
\STATE HeNB-GW broadcasts $\Gamma_{B,m}^{n}$ and $\boldsymbol{\eta}$ values to all small cells
\REPEAT
\FOR{$s=1:S$} 
\STATE All small cells solve (\ref{lagdulfns}) in parallel
\ENDFOR
\STATE All MUEs estimate interference levels on allocated sub-channels and report them to the macrocell
\STATE Macrocell evaluates the sub-gradient (\ref{subgrad}) and updates $\boldsymbol{\eta}$ using the ellipsoid method
\STATE Macrocell sends updated $\boldsymbol{\eta}$ to HeNB-GW
\STATE HeNB-GW broadcasts updated $\boldsymbol{\eta}$ to all small cells
\STATE Macrocell sets $l=l+1$
\UNTIL{Convergence or $l=L_{max}$}
\end{algorithmic}
 \label{Alg_1}
\end{algorithm}

\section{Numerical Results and Discussions}
\label{sec:numres}

\subsection{Parameters}

We evaluate the system performance through extensive simulations under various topologies and scenarios. We have a macrocell located at the origin with radius $300$ m. A hot spot of small cells exists at a distance of $100$ m from the macrocell. The MUEs exist outdoor in this hot spot and are served by the macrocell. Each small cell has $2$ indoor SUEs located randomly on a circular disc around the small cell with an inner radius of $3$ m and an outer radius of $10$ m \cite{3gpp.36.814}. The macrocell has a total power budget of $P_{B, max}=20$ W. 

To model the propagation environment, the channel models from \cite{3gpp.36.814} are used. The channel gains include path-loss, log-normal shadowing, and multipath Rayleigh fading. The path-loss between a small cell and its served SUE,  $PL=38.46+20\log R$ and the path-loss between a small cell and the outdoor MUEs,  $PL=\max(38.46+20\log R, 15.3+37.6\log R)+L_{ow}$, where $R$ is the distance between a small cell and the UE and $L_{ow}$ accounts for losses due to walls. For path-loss between the macrocell and an SUE existing indoor, $PL=15.3+37.6\log R+L_{ow}$ and for path-loss between the macrocell and its MUE, $PL=15.3+37.6\log R$. We have the following values for the standard deviation of log-normal shadowing: $4$ dB for shadowing between SUE and its small cell, $8$ dB for shadowing between MUE and small cell and $10$ dB for shadowing between macrocell and SUE or MUE. The Rayleigh fading gain is modeled as an exponential random variable with unit mean. We assume $\Delta f=180$ KHz, $\epsilon=\frac{0.9}{1+SN}$, and noise power, $N_{o}=10^{-13}$ W. $I_{max}$ is set to any arbitrary large number. All the rate requirements in the numerical results are specified in terms of spectral efficiency (bps/Hz).

In the numerical results, the following performance metrics are used:

\begin{itemize}

\item Average percentage of admitted SUEs  $=\frac{\sum_{s=1}^{S}\sum_{f \in \mathcal{F}_s}^{ }y_{s,f}}{F}\times 100$.

\item Average percentage of channel usage $=\frac{\sum_{s=1}^{S}\sum_{f \in \mathcal{F}_s}^{ }\sum_{n=1}^{N}\Gamma _{s,f}^{n}}{SN}\times 100$.

\end{itemize}

\subsection{Numerical Results}

\subsubsection{Comparison between the traditional and the proposed RA method for macrocell}

In this section, we compare  the two proposed schemes for RA in the macrocell, namely, the formulation in (\ref{macro_pb3}), which we refer to as ``proposed" and the formulation in (\ref{macro_pb5}), which we refer to as ``traditional". Fig. \ref{gains} shows the channel gain realizations for a snapshot of 3 MUEs with 10 sub-channels.

\begin{figure}[th]
\begin{center}
\includegraphics[width=3.3 in]{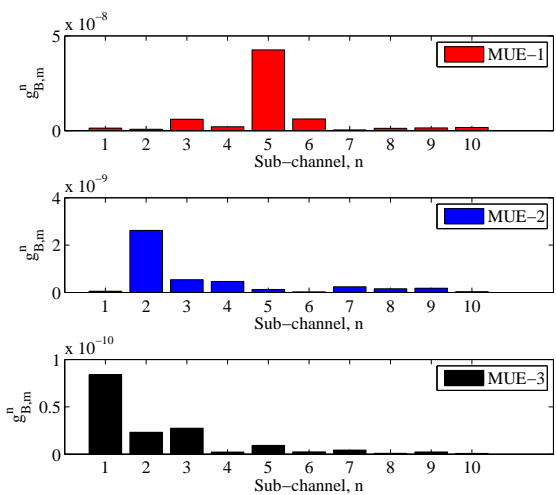}
\caption [c]{Channel gains $g_{B,m}^{n}$ for MUEs $\{1,2,3\}$.}
\label{gains}
\end{center}
\end{figure}

Figs. \ref{allocn_trad}-\ref{allocn_prop} compare the two RA results for the given snapshot with $R_m=5$ bps/Hz. Each figure shows the allocated power $P_{B,m}^{n}$ by the macrocell and the maximum tolerable interference level $I_{m}^{n}$ on allocated sub-channel $n$ to MUE $m$. No values for power $P_{B,m}^{n}$ on the $x$-axis indicate unallocated sub-channel with the corresponding value for $I_{m}^{n}$ set to $ I_{max}$, which means that this sub-channel can be used by the small cell tier  unconditionally. For further clarification, Table \ref{table2} shows the absolute values of $P_{B,m}^{n}$ and $I_{m}^{n}$.

\begin{figure}[th]
\begin{center}
\includegraphics[width=3.3 in]{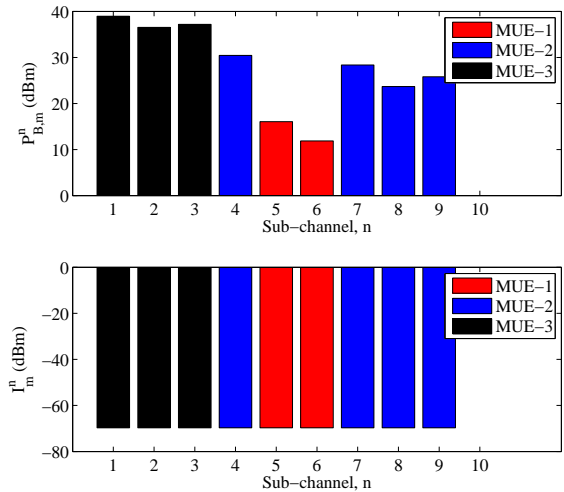}
\caption [c]{Allocated power $P_{B,m}^{n}$ and maximum tolerable interference level $I_{m}^{n}$ for MUEs $\{1,2,3\}$ using the traditional scheme.}
\label{allocn_trad}
\end{center}
\end{figure}

\begin{figure}[th]
\begin{center}
\includegraphics[width=3.3 in]{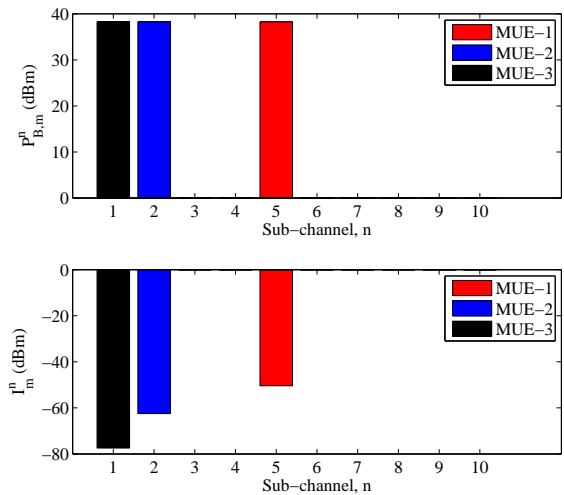}
\caption [c]{Allocated power $P_{B,m}^{n}$ and maximum tolerable interference level $I_{m}^{n}$ for MUEs $\{1,2,3\}$ using the proposed scheme.}
\label{allocn_prop}
\end{center}
\end{figure}

It is clear from Fig. \ref{allocn_trad} that most of the sub-channels are allocated to the MUEs (9 sub-channels out of 10 are allocated to the 3 MUEs), when using the traditional scheme for RA. We notice also that the macrocell favors good sub-channels as they require less transmit power to achieve the rate requirements for the MUEs, leading at the end to minimum transmit power requirements.

Fig. \ref{allocn_prop}, on the other hand, shows that the 3 MUEs require only 3 sub-channels to achieve their rate requirements, as was proved before, using the proposed scheme. Again, the macrocell allocates the best sub-channels to the MUEs. From Figs. \ref{allocn_trad}-\ref{allocn_prop} and Table \ref{table2}, we can notice that the entire power budget of macro BS, $P_{B, max}$, is used in both cases. It is worth mentioning that when we use the traditional scheme for macrocell resource allocation, it does not necessarily mean that it will consume less power than the proposed scheme, since the maximum tolerable interference level $I_{m}^{n}$ is adjusted according to \textbf{Algorithm \ref{Alg_2}} by the macrocell to use the entire power budget. It rather means that, given the maximum tolerable interference levels, the resulting sub-channel and power allocation for the traditional scheme will consume the minimum power and any other allocation will consume a higher power. 

Now, for the maximum tolerable interference levels $I_{m}^{n}$, it is obvious from Figs. \ref{allocn_trad}-\ref{allocn_prop} and Table \ref{table2} that the proposed scheme can sustain higher interference levels from the small cell tier. For the traditional scheme, the sum of the tolerable interference  is given by: $(9 \times 1.0709 \times 10^{-10}) + I_{max}$. On the other hand, the sum of the tolerable interference levels for the proposed scheme can be given by: $(0.1796\times 10^{-10}) + (5.6298\times 10^{-10}) + (91.604\times 10^{-10}) + (7 \times I_{max})$. It is of interest to compare the effect of the two different RA schemes for the macrocell on the small cell tier.


\begin{table*}[th]
\centering
\caption{Absolute Values of $P_{B,m}^{n}$ and $I_{m}^{n}$ for the Traditional and Proposed Macrocell RA Schemes}
\scriptsize
\begin{tabular}{c c c *{10}{|c}|}
\cline{4-13}
& & & \multicolumn{10}{|c|}{Sub-channel\#} \\
\cline{4-13}
& & & 1 & 2 & 3 & 4 & 5 & 6 & 7 & 8 & 9 & 10 \\
\hline \hline
\multicolumn{1}{ |c| }{\multirow{6}{*}{Traditional }} & \multirow{3}{*}{$P_{B,m}^{n}$ (W)}  & \multicolumn{1}{ |c| }{MUE-1} & 0                &	0 &	0 &	0	& 0.0402 & 0.0153 &	0 &	0 &	0 &	0 \\
\cline{3-13}
\multicolumn{1}{ |c| }{\multirow{6}{*}{ scheme}}      &                                                                & \multicolumn{1}{ |c| }{MUE-2} & 0               &	0 &	0 &	1.1097	& 0 &	0 &	0.6855 &	0.2333 & 0.3805 &	0  \\
\cline{3-13}
\multicolumn{1}{ |c| }{}					&						  & \multicolumn{1}{ |c| }{MUE-3} & 7.8381       &	4.4890 &	5.2142 &	0 &	0 &	0 &	0 &	0 &	0 &	0 \\
\cline{2-13}
\multicolumn{1}{ |c| }{}					& \multirow{3}{*}{$I_{m}^{n}$ (W)}      & \multicolumn{1}{ |c| }{MUE-1} & $I_{max}$ &	$I_{max}$ &	$I_{max}$ &	$I_{max}$ &	1.0709 &	1.0709 &	$I_{max}$ &	$I_{max}$ &	$I_{max}$ &	$I_{max}$ \\
\cline{3-13}
\multicolumn{1}{ |c| }{}					& \multirow{3}{*}{$\times 10^{-10}$}     & \multicolumn{1}{ |c| }{MUE-2} & $I_{max}$ &	$I_{max}$ &	$I_{max}$ &	1.0709 &	$I_{max}$ &	$I_{max}$ &	1.0709 &	1.0709 &	1.0709 &	$I_{max}$ \\
\cline{3-13}
\multicolumn{1}{ |c| }{}                			&						   & \multicolumn{1}{ |c| }{MUE-3} & 1.0709       &	1.0709 &	1.0709 &	$I_{max}$ &	$I_{max}$ &	$I_{max}$ &	$I_{max}$ &	$I_{max}$ &	$I_{max}$ &	$I_{max}$ \\
\hline \hline
\multicolumn{1}{ |c| }{\multirow{6}{*}{Proposed }}  & \multirow{3}{*}{$P_{B,m}^{n}$ (W)}   & \multicolumn{1}{ |c| }{MUE-1} & 0                 &	0 &	0 &	0 &	6.6667 &	0 &	0 &	0 &	0 &	0\\
\cline{3-13}
\multicolumn{1}{ |c| }{\multirow{6}{*}{ scheme}}     &						   & \multicolumn{1}{ |c| }{MUE-2} & 0                &	6.6667 &	0 &	0 &	0 &	0 &	0 &	0 &	0 &	0 \\
\cline{3-13}
\multicolumn{1}{ |c| }{}				          &						   & \multicolumn{1}{ |c| }{MUE-3} & 6.6667       &	0 &	0 &	0 &	0 &	0 &	0 &	0 &	0 &	0 \\
\cline{2-13}
\multicolumn{1}{ |c| }{}				          & \multirow{3}{*}{$I_{m}^{n}$ (W)}        & \multicolumn{1}{ |c| }{MUE-1} & $I_{max}$ &	$I_{max}$ &	$I_{max}$ &	$I_{max}$ &	91.604 &	$I_{max}$ &	$I_{max}$ &	$I_{max}$ &	$I_{max}$ &	$I_{max}$ \\
\cline{3-13}
\multicolumn{1}{ |c| }{}					&  \multirow{3}{*}{$\times 10^{-10}$}     & \multicolumn{1}{ |c| }{MUE-2} & $I_{max}$ &	5.6298 &	$I_{max}$ &	$I_{max}$ &	$I_{max}$ &	$I_{max}$ &	$I_{max}$ &	$I_{max}$ &	$I_{max}$ &	$I_{max}$ \\
\cline{3-13}
\multicolumn{1}{ |c| }{}					&						    & \multicolumn{1}{ |c| }{MUE-3} & 0.1796       &	$I_{max}$ &	$I_{max}$ & 	$I_{max}$	& $I_{max}$ &	$I_{max}$ &	$I_{max}$ &	$I_{max}$ &	$I_{max}$ &	$I_{max}$  \\
\hline
\end{tabular}
\label{table2}
\end{table*}

Fig. \ref{AdmittedvsMUEs} compares the average percentage of admitted SUEs when the macrocell performs RA  in order to minimize the sum-power  (as in (\ref{macro_pb5})) and  maximize the sum of tolerable interference levels  (as in (\ref{macro_pb3})) with two different wall loss scenarios. We have the following scenario: $2$ small cells located at $(-10, -100), (10, -100)$ in a square area hot spot of dimensions $20\times20$~m$^2$, $10$ sub-channels, $P_{s, max}=30$ mW, $R_f=50$ bps/Hz, and $R_m=5$ bps/Hz. Numerical results are obtained and averaged for $50$ different realizations, where in each realization, the UE positions and the channel gains are varied. The small cell problem is solved centrally using the convex formulation in (\ref{femto_pbcnvx}). It is clear from the figure that the proposed RA method for the macrocell outperforms the traditional one. When the macrocell performs RA according to the proposed method, it consumes the minimum bandwidth, and therefore, frees as many sub-channels as possible for the small cells. On the other hand, the traditional method consumes more bandwidth than the proposed one, hence, the small cells have more interference constraints to abide by. We also notice that as the wall losses increase, the small cells tend to be more isolated and the impact of resource allocation in the macrocell on the small cell performance is low.

\begin{figure}[th]
\begin{center}
\includegraphics[width=3.3 in]{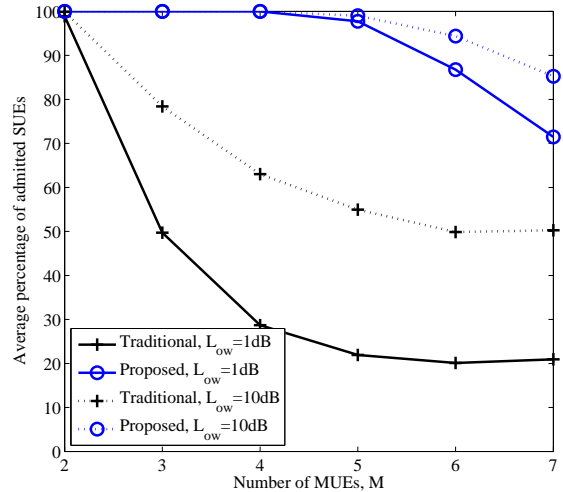}
\caption [c]{Average percentage of admitted SUEs vs. number of MUEs $M$ when the macrocell employs both the proposed and the traditional methods for RA with different wall loss scenarios.}
\label{AdmittedvsMUEs}
\end{center}
\end{figure}

\subsubsection{Comparison between the different formulations for the RA problem for small cells}

Fig. \ref{obj_schemes} compares the values of the objective function for the MINLP formulation in (\ref{femto_pb1}), the centralized convex formulation in (\ref{femto_pbcnvx}), and the distributed formulation in (\ref{lag_femto_pb}) for a snapshot of the following scenario: $2$ small cells located at $(-10, -100), (10, -100)$ in a square area hot spot of dimensions $20\times20$~m$^2$, $3$ sub-channels, $3$ MUEs, $P_{s, max}=30$ mW, $L_{ow}=$ 1 dB, and $R_f=5$ bps/Hz. As was stated previously, the convex formulation provides an upper bound for the solution of the MINLP formulation. Also, we notice that the centralized and distributed formulations have the same solution due to the convexity of the centralized formulation in  (\ref{femto_pbcnvx}). It is worth mentioning that the convex formulation exhibits a  behavior similar to the MINLP formulation. Hence, solving the convex formulations reveals insights into the behavior of the solution of the MINLP formulation. We also notice  that as $R_m$ increases, the  interference constraints for the MUEs become tighter. Hence, the average number of admitted SUEs decreases. Since the objective function in our formulation gives more priority to admission control, the value of objective function decreases with increasing $R_m$.

\begin{figure}[th]
\begin{center}
\includegraphics[width=3.3 in]{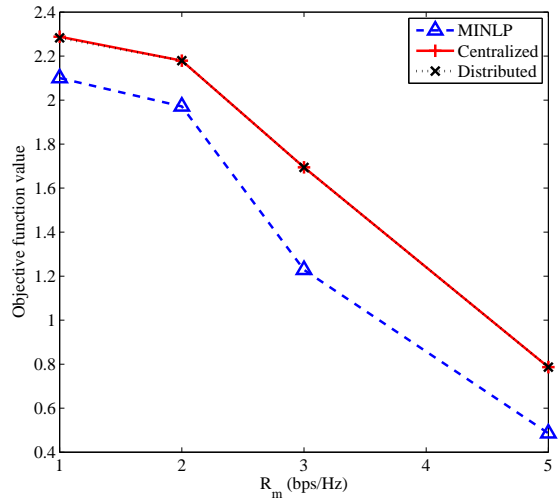}
\caption [c]{The  values of objective function for different formulations vs. $R_m$.}
\label{obj_schemes}
\end{center}
\end{figure}

\begin{figure}[th]
\begin{center}
\includegraphics[width=3.3 in]{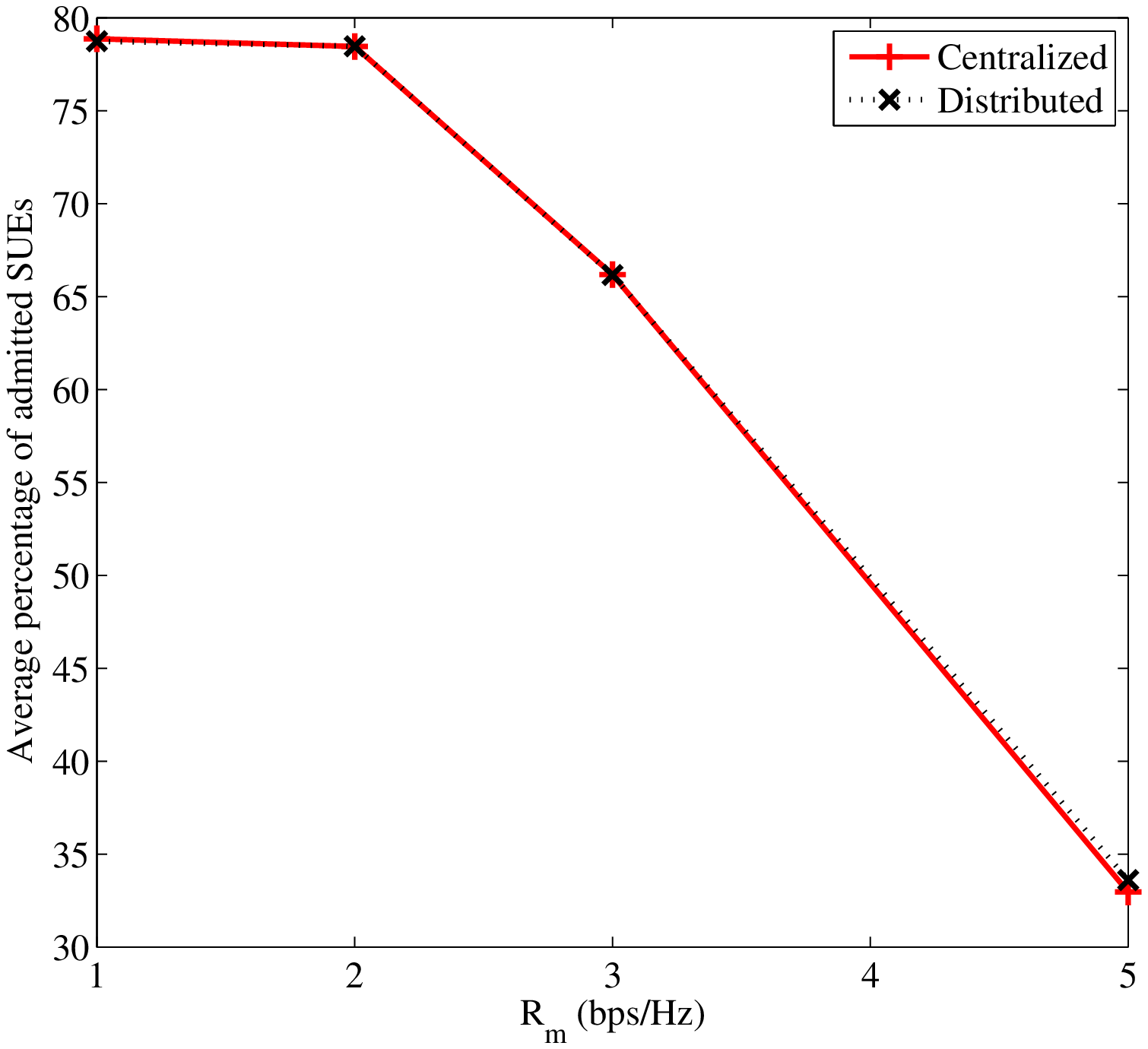}
\caption [c]{Average percentage of admitted SUEs vs. $R_m$.}
\label{admittedRm}
\end{center}
\end{figure}

\begin{figure}[th]
\begin{center}
\includegraphics[width=3.3 in]{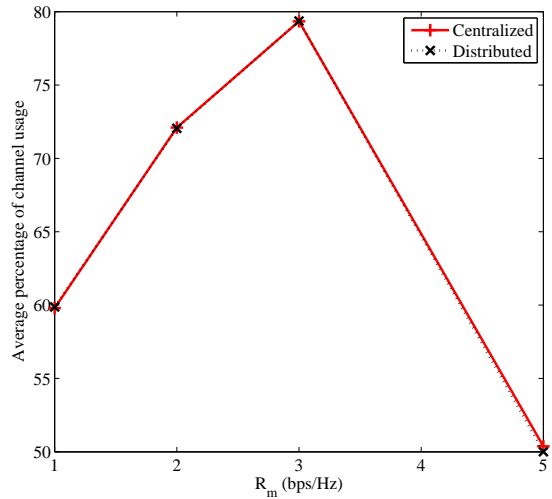}
\caption [c]{Average percentage of channel usage vs. $R_m$.}
\label{chusageRm}
\end{center}
\end{figure}

Figs. \ref{admittedRm}-\ref{chusageRm} show the average percentage of admitted SUEs and channel usage in a small cell vs. $R_m$ for the same scenario considered in Fig. \ref{obj_schemes}. As was discussed in Fig. \ref{obj_schemes}, as the rate requirements for the MUEs increase, they have tighter interference constraints. Hence, the percentage of admitted SUEs generally decreases. We notice in Fig. \ref{admittedRm} that, initially, the average percentage of admitted SUEs is almost constant due to the increased number of used sub-channels as shown in Fig. \ref{chusageRm}. As the MUEs' rate requirements increase further, the increase in the number of used sub-channels is not enough to accommodate the rate requirements of the SUEs, hence, the average percentage of admitted SUEs decreases. 

\subsubsection{Convergence behavior}

Using the same scenario described for the previous figure, Fig. \ref{convergence} shows the convergence behavior of \textbf{Algorithm \ref{Alg_1}}, where the upper bound refers to (\ref{lagdualpb}) and the lower bound refers to the feasible objective obtained by the procedure at the end of Section \ref{sec:DO}. In the figure, the best lower bound is obtained by keeping track of the best primal feasible objective resulting through iterations. It is clear that \textbf{Algorithm \ref{Alg_1}} converges to the optimal solution of (\ref{femto_pbcnvx}) within a few iterations.

\begin{figure}[th]
\begin{center}
\includegraphics[width=3.3 in]{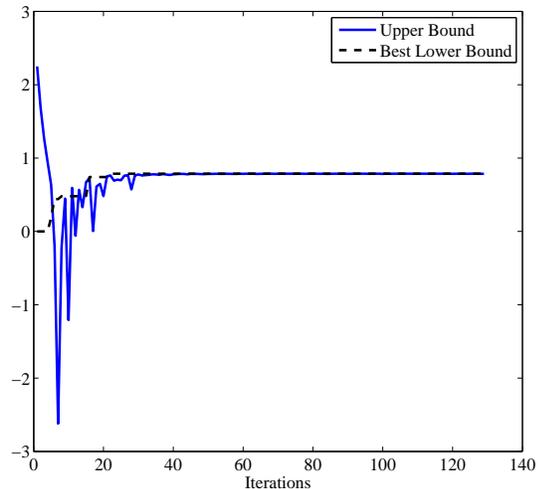}
\caption [c]{Convergence of \textbf{Algorithm \ref{Alg_1}}.}
\label{convergence}
\end{center}
\end{figure}

\subsubsection{Average percentage of admitted SUEs  vs. $R_f$}

In this scenario, we have the following setup: $5$ small cells located at $(-20, -100), (-20, -140), (20, -140), (20, -100), (0, -120)$ in a square area hot spot of dimensions $40\times40$~m$^2$, $5$ sub-channels, $5$ MUEs, $P_{s, max}=30$ mW, $L_{ow}=$ 1 dB and $R_m=4$ bps/Hz. Numerical results are obtained and averaged for $50$ different realizations, where in each realization, the UE positions and channel gains are varied. Fig. \ref{admittedRf} shows the average percentage of admitted SUEs  vs. $R_{f}$. We notice that, generally, as the rate requirement increases, more SUEs are in outage. We also notice that the distributed scheme converges approximately to the same solution as the centralized solution.

\begin{figure}[th]
\begin{center}
\includegraphics[width=3.3 in]{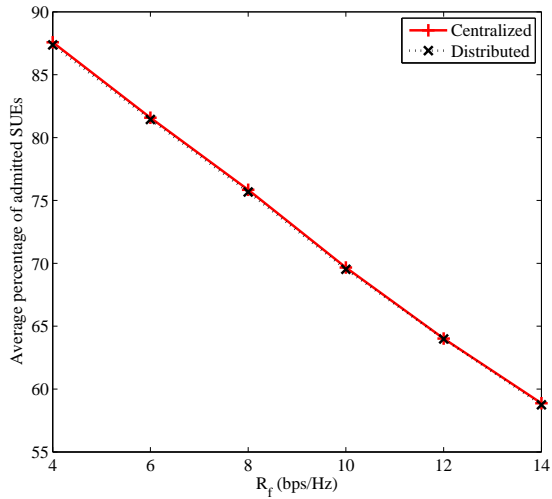}
\caption [c]{Average percentage of admitted SUEs vs. $R_f$.}
\label{admittedRf}
\end{center}
\end{figure}

\subsubsection{Average percentage of admitted SUEs vs. $P_{s, max}$}

We have the same setup as the one for the previous figure except for $R_f=10$ bps/Hz. Fig. \ref{admitvspsmax} shows the average percentage of admitted SUEs  vs. $P_{s, max}$. We notice that as the maximum transmit power of the small cells increases, the average number of admitted SUEs increases. This rate of increase, however, is not fixed as the system is limited by the interference constraints for MUEs.

\begin{figure}[th]
\begin{center}
\includegraphics[width=3.3 in]{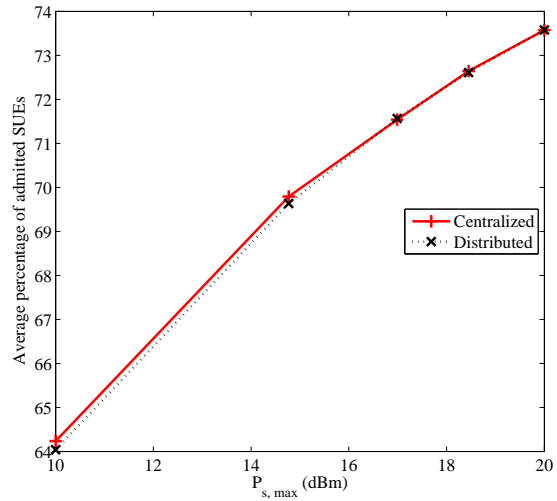}
\caption [c]{Average percentage of admitted SUEs vs. $P_{s, max}$.}
\label{admitvspsmax}
\end{center}
\end{figure}

\subsection{Summary of Major Observations}

The major observations from the numerical analysis can be summarized as follows:

\begin{itemize}

\item In a multi-tier network, it is critical to consider the impact of RA decisions in one tier on the other one. For the macrocell network, as different RA schemes are used to achieve the same rate requirements for the MUEs, they affect the performance of the small cell tier differently.

\item Additional network tiers can be accommodated by minimizing the used bandwidth. The proposed problem formulation for resource allocation in the macrocell leads to a minimal use of the system bandwidth which allows to admit a higher number of SUEs. 

\item It is foreseen that, employing other schemes, such as a one to ``maximize the sum-rate'', for either the macrocell tier or the small cell tier, is not a good option as it will not be able to maximally accommodate additional tiers. Similar to the ``minimize the total sum-power" scheme, the ``maximize the sum-rate'' scheme is known to consume most of the available bandwidth \cite{Kibeom2006}.

\item For a given macrocell RA policy, increasing the rate requirements for the MUEs degrades the performance of small cell tier in terms of the average number of admitted SUEs.

\item By exploiting the time sharing property (i.e., the UEs time share the sub-channels when served by the small cells), a convex optimization formulation can be developed for the RA and AC problem for the small cells. This convex formulation enables us to solve the problem efficiently in a distributed fashion. The distributed algorithm for resource allocation for the small cell tier converges to the same solution as the centralized solution.

\item If the deployment of the small cells is such that they are well isolated, resource allocation in the macrocell might have very small effect on the performance of small cells.



\end{itemize}

\section{Conclusion}
\label{sec:conc}

We have proposed a complete framework for the resource allocation and admission control problem in a two-tier OFDMA cellular network.  Different optimization problems with new objectives have been formulated for the macrocell tier and the small cell tier. The macrocell tier aims at allocating resources to its MUEs in a way that can tolerate the maximum possible interference from the small cell tier. This problem has been shown to be an MINLP. However, we have proved that the macrocell can solve another alternate optimization problem that yields the optimal solution in polynomial time. Now, given the interference constraints for the MUEs, the small cells perform resource allocation and admission control with the objective of maximizing the number of admitted SUEs and serving them with the minimum possible bandwidth. This problem has also been shown to be an MINLP. A convex relaxation has been used to study the behavior of the MINLP formulation. Since centralized solutions for resource allocation are not practical for dense networks, a distributed solution for resource allocation and admission control has been proposed using dual decomposition technique and has been shown to converge to the same solution as the centralized one. Numerical results have shown the significane of tier-aware resource allocation methods. 


\bibliographystyle{ieeetr}

\end{document}